\documentclass[spconf, 10pt]{IEEEtran}
\usepackage{tikz}
\usepackage{pgfplots}
\usepackage{amsmath,epsfig,amssymb,algorithm,algpseudocode,amsthm,cite,url}
\usepackage{hyperref}
\usepackage{enumerate}
\usepackage[latin1]{inputenc}
    \DeclareMathOperator{\tr}{tr}

\theoremstyle{plain}
\newtheorem{thm}{Theorem}

\newtheorem{prop}[thm]{Proposition}

\theoremstyle{definition}

\theoremstyle{remark}

\newtheoremstyle{specialcasestyle}{1mm}{1mm}{\upshape}{}{\bfseries\upshape}{.}{0mm}{}
\theoremstyle{specialcasestyle}

  \newcommand{\bT}{{\bf T}}

   \newcommand{\bR}{{\bf R}}

  \newcommand{\bA}{{\bf A}}

    \newcommand{\bI}{{\bf I}}
      \newcommand{\bh}{{\bf h}}
      \newcommand{\bH}{{\bf H}}

      \newcommand{\bx}{{\bf x}}
         \newcommand{\bu}{{\bf u}}

          \newcommand{\bz}{{\bf z}}
      
       \newcommand{\by}{{\bf y}}
       \newcommand{\bb}{{\bf b}}

      \newcommand{\bQ}{{\bf Q}}
  
      \newcommand{\bn}{{\bf n}}
      
     \newcommand{\minimize}[1]{{\underset{{#1}}{\mathrm{minimize}}}}

\usepackage{graphicx}
\usepackage{color}

\begin{document}
\title{
Precoding Design for Single-RF Massive MIMO systems: A Large System Analysis (Extended Version)
}
\author{\IEEEauthorblockN{Houssem Sifaou, Abla Kammoun and Mohamed-Slim Alouini}
\\
\IEEEauthorblockA{
Computer, Electrical and Mathematical Sciences and Engineering, King Abdullah University of Science and Technology, Thuwal, Saudi Arabia\\}
}

\maketitle
\begin{abstract}
This work revisits a recently proposed precoding design for massive multiple-input multiple output (MIMO) systems that is based on the use of an instantaneous total power constraint. The main advantages of this technique lie in its suitability to single RF MIMO systems coupled with a very-high power efficiency. 
Such features have been proven using simulations for uncorrelated channels. Based on tools from random matrix theory, we propose in this work to analyze the performance of this precoder for more involved channels accounting for spatial correlation. The obtained expressions are then optimized in order to maximize the signal-to-interference-plus-noise ratio (SINR). Simulation results are provided in order to illustrate the performance of the optimized precoder in terms of
peak-to-average power ratio (PAPR) and  signal-to-interference-plus-noise ratio (SINR).
\end{abstract}


\IEEEpeerreviewmaketitle

\section{Introduction}
Large scale multiple-input multiple-output systems (MIMO), also known as massive MIMO systems, are considered as a promising technology in modern wireless networks\cite{Marzetta2010a, Hoydis2013a, Rusek2013a, Larsson2014}. By using a very large number of antennas and intelligent precoding schemes at the base station (BS), these systems can focus energy in small regions thereby reducing the inter-user interference and at the same time the energy consumption.

One important challenge in massive MIMO systems is to deal with  multiple-access interference (MAI). For downlink transmissions, MAI mitigation can be accomplished at the BS using precoding techniques.
The precoding design can be based on  average or instantaneous constraints. The use of instantaneous constraints is more interesting from a practical standpoint,  as it allows  to take into account the presence of power amplifiers. If each antenna is fed by a separate power amplifier, then a per-antenna power constraint should be used. In case  a single power amplifier is employed for all antennas, implying the use of a single radio frequency (RF) chain, the use of a total power constraint
becomes more appropriate.  

The latter case has recently been studied in  \cite{Sadaghat2014a},
wherein the proposed scheme was shown to exhibit a higher power efficiency as compared to its predecessors, in addition to a lower peak-to-average ratio. 
Nevertheless,  being based on the use of an instantaneous power constraint, the resulting precoder involves solving a fixed point equation, which does not facilitate the carrying out of performance analysis. This motivates our work. In particular, using tools from random matrix theory, we analyze the asymptotic performance of the scheme of  \cite{Sadaghat2014a}. This allows us to accurately approximate its performances under more involved channels considering the spatial
correlation, but more importantly, it provides us insights into the asymptotic behavior of the parameters intervening in the precoder of \cite{Sadaghat2014a}.

The remainder of the paper is organized as follows. The next section introduces the system model and formulates the problem. In section III, the optimal parameters of the precoding design are determined in the asymptotic regime. Simulations results are presented in section IV while some conclusions and implications are drawn in section V.

\section{System Model and Problem Formulation}
We consider a single cell massive MIMO system where a BS equipped with $M$ antennas communicates with $K<M$ single antenna UEs. The channel vector is modeled as
\begin{equation}
\bh_k=\bR^{\frac{1}{2}}\bz_k,\label{channel_model}
\end{equation}
where $\bR\in \mathbb{C}^M\times \mathbb{C}^M$ is the spatial correlation matrix and $\bz_k \sim \mathcal{CN}({\bf 0},{\bf I})$. We assume that $\bR$ is the same for all the UEs, an assumption which has been considered in several previous works \cite{Nam2014a}. For notational convenience, we define $c=\frac{K}{M}$.

The received signal vector $\by$ at all the UEs is: 
\begin{align}
\by=\bH^{\mbox{\tiny H}}\bx+\bn,
\end{align}
where $\bx$ is the precoded vector at the BS, $\bH=[\bh_1,\cdots,\bh_K]$ is the channel matrix and $\bn\sim\mathcal{CN}({\bf 0},\sigma^2{\bf I})$ stands for the additive white Gaussian noise. We assume that the BS performs precoding, implying that ${\bf x}$ is the precoding output of a data vector ${\bf u}$, before being transmitted via the channel.  We assume that the BS employs Gaussian codebooks, that is, ${\bf u}\sim\mathcal{CN}({\bf 0},{\bf I})$.
The instantaneous transmit power at the BS is given by
$$
P=\bx^{\mbox{\tiny H}}\bx.
$$
The output vector ${\bf x}$ of the precoding is designed in such a way to minimize a squared distortion measure. In particular, it is given as the solution of the following problem:
\begin{equation}
\begin{aligned}
\mathcal{P}:  \,\, \, &\minimize {\bx} \, \, \,\, \, \, \|\bH^{\mbox{\tiny H}}\bx-\sqrt{\gamma}\bu\|^2 \\
 \,\, \,&\mathrm{subject \,\, to} \, \, \,\, \, \,  \bx^{\mbox{\tiny H}}\bx\le P_a,
\end{aligned}
\label{eq:prob}
\end{equation}
where $P_a$ is the maximum allowed transmit power and $\gamma$ is a design parameter which will be optimized later.
The main difference between the precoding  solving \eqref{eq:prob} and the classical precoding techniques proposed  within the framework of  massive MIMO systems lies in the use of an instantaneous power constraint instead of an average power constraint. If the average power constraint were used, then the precoding would reduce to the ZF precoding. The precoding solving \eqref{eq:prob} has been determined in \cite{Sadaghat2014a}, where it has been proven that the solution is
as follows:  
Define
$$
\phi (\bx)= \| \bH^{\mbox{\tiny H}} \bx - \sqrt{\gamma}\bu\|^2=(\bx-\bb)^{\mbox{\tiny H}}\bA(\bx-\bb),
$$
where $\bb=\sqrt{\gamma} \bH(\bH^{\mbox{\tiny H}}\bH)^{-1}\bu$  and $\bA=\bH\bH^{\mbox{\tiny H}}$. Vector $\bb$ corresponds to the zero forcing (ZF) solution.
As reported in \cite{Sadaghat2014a}, by applying the method of Lagrange multipliers, it can be shown that $\bx$ is a stationary point of $\phi (\bx)$ if and only if 
$$
\bA(\bx-\bb)=- \delta \bx,
$$
where $\delta\ge 0$ is the Lagrange multiplier. Let ${\bf x}^{*}$ denotes the optimal solution of \eqref{eq:prob}. Then, two cases should be considered:

{\bf Case 1:} If $\bb^{\mbox{\tiny H}}\bb \le P_a$, then $\delta=0$ and $\bx^\star=\bb$.

{\bf Case 2:} If $\bb^{\mbox{\tiny H}}\bb > P_a$, then $\delta>0$ and $\bx^\star=(\bA + \delta \bI)^{-1} \bA \bb$, where $\delta^\star$ is selected so that $ \left(\bx^\star\right)^{\mbox{\tiny H}}\bx^\star=P_a$. 

This shows that depending on $P_a$, the optimal precoding can be either the ZF precoder or the regularized zero-forcing precoder. The solution of \eqref{eq:prob} can be thought of as a clipping algorithm, returning the ZF precoding if this latter exhibits a power less than $P_a$ or the RZF precoding otherwise. When the RZF is used, a scalar $\delta^\star$ must be selected so that the total transmit power is equal to $P_a$. The work in \cite{Sadaghat2014a} suggests to determine
$\delta^\star$ numerically, which does not give insights on the optimal choice of $\gamma$. To overcome this issue, we resort to tools from random matrix theory and assume that $M$ and $K$ increase with the same pace. This allows us to compute a deterministic equivalent approximating the SINR. Maximizing the obtained expression, an optimal value for $\gamma$ can be thus derived.  More details in this respect will be provided in the next section.  


\section{Asymptotic analysis}

\subsection{Deterministic equivalent of $\delta^\star$}
In this section, we will only treat the case in which ${\bf b}^{H}{\bf b}>P_a$ since, otherwise, the lagrangian multiplier $\delta$ is known to be equal to zero. If   ${\bf b}^{H}{\bf b}>P_a$, the transmit power can be expressed as:
\begin{equation}
P( \delta)=\bx^{\mbox{\tiny H}}\bx=\gamma \bu^H (\bH^{\mbox{\tiny H}} \bH+\delta \bI)^{-1} \bH^{\mbox{\tiny H}} \bH (\bH^{\mbox{\tiny H}} \bH+  \delta \bI)^{-1} \bu,
\label{P}
\end{equation}
In the first step, we compute a deterministic equivalent of the transmit power for a fixed $\delta>0$. For notational convenience, we define $\rho=\frac{\delta}{K}$.
\begin{prop}
    Let $\alpha(t)$ be the unique solution to the following fixed-point equation:
\begin{equation}
\alpha(t)=\frac{1}{K}\tr\left(\bR\left(\bI_M+ \frac{t\bR}{1+t\alpha(t)}\right)^{-1} \right).
\end{equation}

In the asymptotic regime, we have 
$$
 P (\rho)-\overline P(\rho)\xrightarrow[K\to \infty]{a.s.}0,
$$
 where
\begin{equation}
\overline P(\rho)=\frac{\gamma}{\rho^2} \beta(\frac{1}{\rho}),
\end{equation}
and
\begin{equation}\beta(t) =\frac{\frac{1}{K}\tr\left(\bR \bT^2(t)\right)}{(1+t\alpha(t))^2-\frac{t^2}{K}\tr\left(\bR \bT(t)\bR \bT(t)\right)} .\end{equation}
with $\bT(t)=\left(\bI_M+ \frac{t\bR}{1+t\alpha(t)}\right)^{-1}$.
\end{prop}
\begin{proof}
The result is obtained by direct application of the result of \cite[Theorem 8]{Kammoun2014a}.
\end{proof}

For fixed $\gamma$ and in the case $\bb^H\bb \ge P_a $, as per the design procedure described in the previous section, the instantaneous optimal $\rho$ is selected so that: 
$$
P(\rho^\star)=P_a.
$$
It can be shown that $\overline P(\rho)$ is a strictly decreasing function. As a consequence, a deterministic equivalent for $\rho^\star$, $\overline{\rho}$, is solution to the following equation:
 \begin{equation}
 \overline P(\rho)=P_a,
\label{delta_gamma1}
 \end{equation}
 It is worth pointing out that the use of $\overline{\rho}$ instead of $\rho^\star$ might be not a good idea in practice, since it might for some channel realizations lead to an instantaneous power that is higher than $P_a$. This value can be, however, leveraged in order to determine the optimal value of $\gamma$ that maximizes the power efficiency and the SINR. 

\subsection{Optimal $\gamma$}
In this section, we consider the SINR of all users defined as\cite{Sadaghat2014a}:
 \begin{equation}
 {\rm SINR} =\frac{\gamma \bu^{\mbox{\tiny H}}\bu}{\| \bH^{\mbox{\tiny H}} \bx - \sqrt{\gamma}\bu\|^2+K\sigma^2}.
 \label{SINR}
 \end{equation}

A deterministic equivalent of the SINR is given in the following proposition.
\begin{prop}
In the asymptotic regime, we have 
$$
{\rm SINR}-\overline{\rm SINR}\xrightarrow[K\to \infty]{a.s.}0,
$$
where
$$
\overline{\rm SINR}=\frac{1}{\frac{\sigma^2}{P_a\rho^2}\beta(\frac{1}{\rho})+\left(\frac{1}{K}\tr \left(\bT(\frac{1}{\rho})\right)-\frac{M-K}{K}-\frac{1}{\rho} \beta(\frac{1}{\rho})\right)},
$$
\end{prop}
\begin{proof}
See appendix \ref{appA}.
\end{proof}
With the asymptotic equivalent of the SINR and the transmit power at hand, we are now in position to determine the optimal $\gamma$:
\begin{prop}
The optimal $\gamma$ that maximizes the SINR in the asymptotic regime is: $$\gamma^\star=\frac{1}{\frac{P_a}{\sigma^4}\beta\left(\frac{P_a}{\sigma^2}\right)}.$$
\end{prop}
\begin{proof}
See appendix \ref{appB}.
\end{proof}

\subsection{Particular case: $\bR=\bI_M$}
The treatment of $\bR=\bI_M$ is detailed hereafter. Interestingly, it turns out that in this case, the transmit power and the SINR admit explicit closed-form expressions that do not required solving fixed-point equations. 
\begin{equation}
P( \delta)=\bx^{\mbox{\tiny H}}\bx=\gamma \bu^H (\bH^{\mbox{\tiny H}} \bH+\delta \bI)^{-1} \bH^{\mbox{\tiny H}} \bH (\bH^{\mbox{\tiny H}} \bH+  \delta \bI)^{-1} \bu,
\label{P}
\end{equation}
\begin{prop}
In the asymptotic regime, we have 
$$
 P (\rho)-\overline P (\rho)\xrightarrow[K\to \infty]{a.s.}0,
$$
 where
$$
\overline P (\rho)=\gamma [m(\rho)+ \rho m'(\rho)]\overset{\Delta}{=}\gamma \bar f(\rho),
$$
and $m(\rho) =\frac{-2}{1-c-\rho-\sqrt{(1-c+\rho)^2+4c\rho}} $.
\end{prop}
\begin{proof}
See Appendix \ref{appC}.
\end{proof}
  
\begin{prop}
\label{prop2}
In the asymptotic regime, we have
$$
{\rm SINR}-\overline{\rm SINR}\xrightarrow[K\to \infty]{a.s.}0,
$$
where
 \begin{equation}
\overline{\rm SINR}=\frac{1}{ -\bar \rho^2 m'(\bar\rho)+\frac{\sigma^2}{P_a}\bar f(\bar\rho)}.
\label{delta_gamma}
 \end{equation}
\end{prop}
\begin{proof}
See Appendix \ref{appD}.
\end{proof}
From equation \eqref{delta_gamma}, the optimal $\gamma$ that maximizes the SINR can be determined by optimizing over $\bar \rho$.
Note that function $g(\bar\rho)=-\bar \rho^2 m'(\bar \rho)+\frac{\sigma^2}{P_a}\bar f(\bar\rho)$ is convex over $[0,+\infty[$. This can be easily shown by taking the second derivative. Besides, it can be easily checked that $g'(\frac{\sigma^2}{P_a})=0$. Hence, $\bar \rho$ that maximizes the SINR is $\bar\rho^\star=\frac{\sigma^2}{P_a}$ and the optimal $\gamma$ is such as $\gamma^\star=\frac{Pa}{\bar f\left(\frac{\sigma^2}{P_a}\right)}$.

\section{Power Efficiency Analysis}
Assuming that the maximum allowed transmit power is $P_a$, the average power efficiency can be expressed as
$$
\eta_t=\eta_a\frac{\mathbb{E}\left[({\bx^*})^{\mbox{\tiny H}} \bx^*\right]}{P_a},
$$
where $\eta_a$ is the efficiency of the power amplifier \cite{Sadaghat2014a}.
As per the design procedure described in section II, quantity  $({\bx^*})^{\mbox{\tiny H}} \bx^*$ is given by:
$$
({\bx^*})^{\mbox{\tiny H}} \bx^* = {\bf b}^{\mbox{\tiny H}}{\bf b}\boldsymbol{1}_{\left\{{\bf b}^{\mbox{\tiny H}}{\bf b}\leq P_a\right\}} +P_a \boldsymbol{1}_{\left\{{\bf b}^{\mbox{\tiny H}}{\bf b}\geq P_a\right\}}
$$
where ${\bf 1}_{A}$ denotes the indicator function of set $A$. It follows, thus, that the average power efficiency writes as: 
\begin{equation}
\eta_t=\eta_a\left(\int_{\frac{P_a}{\gamma}}^{\infty} f_Z(z)dz+\frac{\gamma}{P_a}\int_{0}^{\frac{P_a}{\gamma}}z f_Z(z)dz\right).
\label{efficiency}
\end{equation}
where $Z\triangleq\frac{1}{\gamma}{\bf b}^{\mbox{\tiny H}}{\bf b}={\bf u}^{\mbox{\tiny H}}\left({\bf H}^{\mbox \tiny H}{\bf H}\right)^{-1}{\bf u}$ and $f_Z$ denotes its probability density function (pdf). Since ${\bf H}$ is unitarily invariant, $Z\stackrel{d}{u}\frac{X}{Y}$ where $X=\frac{1}{K}{\bf u}^{\mbox{\tiny H}}{\bf u}$ and $Y=\frac{1}{K}\frac{1}{\left[\left({\bf H}^{\mbox{\tiny H}}{\bf H}\right)^{-1}\right]_{1,1}}$. 
In case of uncorrelated channels ($\bR=\bI_M$), variable $Z$ is known to follow a scaled $\mathcal{F}$-distribution \cite{Sadaghat2014a,HOC02}. However, when the channels are correlated, finding a closed-form expressions for $Z$ might be not possible, the main difficulty lying in the characterization of the distribution $Y$. To handle this case, we will use the asymptotic approximation developed in \cite{Moustakas2013}. Prior to that, we shall define the following notations.
For $s\geq 0$, we define 
$I_{\rm erg}(s)$ as:
\begin{equation}
I_{erg}(s)=\ln \det \left[\bI_M+\bR \left(s+\frac{1}{t(s)}\right)\right]+K[\ln(t(s))-1],
\label{I_{erg}}
\end{equation}
where $t(s)$  is the unique solution to:
$$
1=\frac{1}{K}\sum_{i=1}^M \frac{R_i}{t(s)(1+R_i s)+R_i}
$$
For $y>0$, define $s_0$ as the unique positive solution in $s$ to the following equation:
$$
1=\frac{1}{K}\sum_{i=1}^M \frac{R_i}{y(1+R_i s)+R_i}
$$
Define $v_1, v_2$ as:
\small
\begin{align*}
v_1&=-\ln|1-M_{t1}M_{r1}|-\ln|1-M_{t2}M_{r2}|+2-\ln|1-M_{t3}M_{r3}|,\\
v_2&=-\ln \left| \frac{M_{r1}}{1-M_{t1}M_{r1}} \right|,
\end{align*}
with
\begin{align*}
M_{t1}&=\frac{1}{[y]^2},\\
    M_{t2}&=\frac{1}{[t(0)]^2},\\
    M_{t3}&=\frac{1}{[t(0) y]},\\
M_{r1}&=\frac{1}{K}\tr\left( \bR^2\left[\bI_M+\bR\left(s_0+\frac{1}{y}\right)\right]^{-2}\right),\\
    M_{r2}&=\frac{1}{K}\tr\left( \bR^2\left[\bI_M+\bR\left(\frac{1}{t(0)}\right)\right]^{-2}\right),\\
    M_{r3}&=\frac{1}{K}\tr\left( \bR^2\left[\bI_M+\bR\left(\frac{1}{t(0)}\right)\right]^{-1}\right.
\left.\left[\bI_M+\bR\left(s_0+\frac{1}{y}\right)\right]^{-1}\right).
\end{align*}
\normalsize
Then, 
up to an error of order $\mathcal{O}(K^{-1})$, the asymptotic pdf of $Y$ is given by \cite{Moustakas2013}: 
\begin{equation}
f_Y(y)=\sqrt{\frac{K}{2\pi}}e^{\left(K s_0 y-I_{erg}(s_0)+I_{erg}(0)+\frac{v_1+v_2}{2}\right)},
\end{equation}
while the asymptotic pdf of $Z$ can be evaluated numerically using the following formula: 
\begin{equation}
f_Z(z)=\int_{0}^{\infty} y f_X(zy)f_Y(y) dy.
\end{equation}
 With the distribution of $Z$ on hand,  the power efficiency $\eta_t$ in \eqref{efficiency} of this precoding scheme  is in its turn computed numerically.  

\section{Numerical results}

In this section, numerical results are presented in order to illustrate the valuable implications of our theoretical results. Moreover, Montecarlo simulations are used to validate the analysis in the asymptotic regime. In all figures, the results are obtained by averaging over $1000$ different channel realizations. In all  simulations, the covariance matrix is modeled as:
 \begin{align}
\left[{\bR}\right]_{i,j}=\left\{\begin{array}{ll} a^{j-i} \hspace{1cm}&\textnormal{if} \hspace{0.2cm} i\leq j\\ (a^{j-i})^*, \hspace{1cm}&\textnormal{if} \hspace{0.2cm} i>j,\end{array}\right.
  \end{align}
where $a \in \mathbb{C}$ is the (complex) correlation coefficient of neighboring antenna and $x^*$ denotes the complex conjugate of $x$. This is known as the exponential correlation model \cite{Loyka2001a}. Without loss of generality, we assume that $P_a=1$. 

We first consider a massive MIMO system with 80 antennas at the BS and 40 UEs. Fig. \ref{fig1} plots the SINR at left Y-axis and the power efficiency $\eta_t$ at the right Y-axis versus $\frac{\frac{1}{c}-1}{\gamma}$ for different values of the noise variance and different covariance matrices. It is important to note from this figure that there exists an optimal $\gamma$ that maximizes the SINR and which depend on the noise variance, the channel covariance matrix and the ratio between $K$ and $M$. Given all the required parameters, we are able to determine the optimal $\gamma$ as shown in the previous section.
Moreover, we observe that the power efficiency increases slightly with the correlation coefficient. However, the increase in power efficiency is negligible compared to the loss on SINR.

In fig. \ref{fig2}, we plot the SINR versus the number of the BS antennas $M$ for a fixed number of UEs $K=40$. We compare the performance of the considered precoding technique when the optimal $\gamma$ is used and when an arbitrary $\gamma$ is chosen. It is observed that our proposed $\gamma$ provides a considerable increase in performance. 

Fig. \ref{fig3} plots the maximum SINR of the single-RF system and the SINR of classical MIMO with MMSE precoding \cite{Peel2005} at left Y-axis versus M when $K=10$ and $\sigma^2=1$. The PAPR of the single-RF transmitter is plotted at the right Y-axis. It is observed that the single-RF transmitter has very low PAPR and exhibits an interesting gain in terms of SINR compared with the optimal MMSE precoding. The proposed precoding reduces slightly the performance while reducing significantly the PAPR.

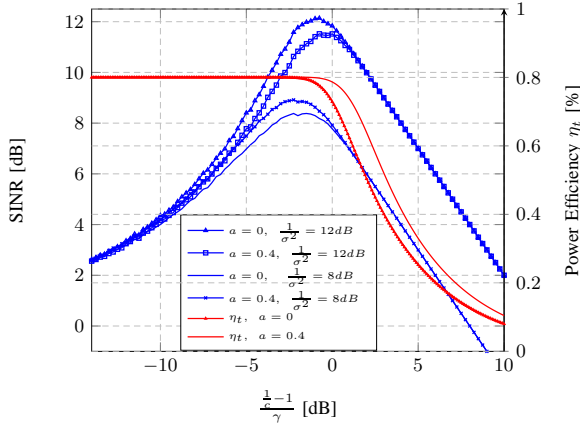
\begin{figure}
  \centering
   \begin{tikzpicture}[scale=0.8,font=\small]
    \renewcommand{\axisdefaulttryminticks}{4}
    \tikzstyle{every major grid}+=[style=densely dashed]
    \tikzstyle{every axis y label}+=[yshift=0pt]
    \tikzstyle{every axis x label}+=[yshift=0pt]
    \tikzstyle{every axis legend}+=[cells={anchor=west},fill=white,
        at={(0.215,0.40)}, anchor=north west, font=\tiny ]
    \begin{axis}[
      xmin=-14,
      ymin=-1,
      xmax=10,
      ymax=12.5,
      grid=major,
      scaled ticks=true,
   			xlabel={$\frac{\frac{1}{c}-1}{\gamma}$ [dB]},
   			ylabel={SINR [dB] }			
      ]
  
          \addplot[color=blue,mark size=1pt,mark=triangle,line width=0.6pt] coordinates{

(-14,2.627551e+00)(-1.375000e+01,2.702377e+00)(-1.350000e+01,2.794968e+00)(-1.325000e+01,2.835663e+00)(-13,2.949810e+00)(-1.275000e+01,3.015732e+00)(-1.250000e+01,3.128137e+00)(-1.225000e+01,3.190814e+00)(-12,3.361038e+00)(-1.175000e+01,3.453107e+00)(-1.150000e+01,3.629699e+00)(-1.125000e+01,3.595979e+00)(-11,3.810819e+00)(-1.075000e+01,3.903030e+00)(-1.050000e+01,4.052537e+00)(-1.025000e+01,4.189719e+00)(-10,4.357549e+00)(-9.750000e+00,4.415375e+00)(-9.500000e+00,4.570210e+00)(-9.250000e+00,4.693467e+00)(-9,4.908842e+00)(-8.750000e+00,5.095705e+00)(-8.500000e+00,5.246731e+00)(-8.250000e+00,5.385481e+00)(-8,5.583286e+00)(-7.750000e+00,5.808814e+00)(-7.500000e+00,6.025694e+00)(-7.250000e+00,6.155019e+00)(-7,6.374731e+00)(-6.750000e+00,6.618178e+00)(-6.500000e+00,6.814166e+00)(-6.250000e+00,6.986377e+00)(-6,7.312581e+00)(-5.750000e+00,7.475089e+00)(-5.500000e+00,7.785337e+00)(-5.250000e+00,8.100760e+00)(-5,8.405279e+00)(-4.750000e+00,8.517550e+00)(-4.500000e+00,8.885108e+00)(-4.250000e+00,9.124589e+00)(-4,9.453797e+00)(-3.750000e+00,9.814213e+00)(-3.500000e+00,1.008634e+01)(-3.250000e+00,1.037806e+01)(-3,1.071506e+01)(-2.750000e+00,1.092318e+01)(-2.500000e+00,1.119505e+01)(-2.250000e+00,1.158225e+01)(-2,1.176761e+01)(-1.750000e+00,1.188325e+01)(-1.500000e+00,1.198100e+01)(-1.250000e+00,1.202983e+01)(-1,1.211986e+01)(-7.500000e-01,1.214033e+01)(-5.000000e-01,1.201135e+01)(-2.500000e-01,1.191572e+01)(0,1.183355e+01)(2.500000e-01,1.163843e+01)(5.000000e-01,1.143117e+01)(7.500000e-01,1.119959e+01)(1,1.097985e+01)(1.250000e+00,1.073814e+01)(1.500000e+00,1.049354e+01)(1.750000e+00,1.024743e+01)(2,9.998598e+00)(2.250000e+00,9.749577e+00)(2.500000e+00,9.499487e+00)(2.750000e+00,9.249578e+00)(3,8.999899e+00)(3.250000e+00,8.750000e+00)(3.500000e+00,8.500000e+00)(3.750000e+00,8.250000e+00)(4,8.000000e+00)(4.250000e+00,7.750000e+00)(4.500000e+00,7.500000e+00)(4.750000e+00,7.250000e+00)(5,7)(5.250000e+00,6.750000e+00)(5.500000e+00,6.500000e+00)(5.750000e+00,6.250000e+00)(6,6.000000e+00)(6.250000e+00,5.750000e+00)(6.500000e+00,5.500000e+00)(6.750000e+00,5.250000e+00)(7,5)(7.250000e+00,4.750000e+00)(7.500000e+00,4.500000e+00)(7.750000e+00,4.250000e+00)(8,4.000000e+00)(8.250000e+00,3.750000e+00)(8.500000e+00,3.500000e+00)(8.750000e+00,3.250000e+00)(9,3)(9.250000e+00,2.750000e+00)(9.500000e+00,2.500000e+00)(9.750000e+00,2.250000e+00)(10,2.000000e+00)
         		};
\addlegendentry{ $a=0 , \ \ \frac{1}{\sigma^2}=12dB$}
                      \addplot[color=blue,mark size=1pt,mark=square,line width=0.6pt] coordinates{

(-14,2.559186e+00)(-1.375000e+01,2.624790e+00)(-1.350000e+01,2.716921e+00)(-1.325000e+01,2.748337e+00)(-13,2.858476e+00)(-1.275000e+01,2.919158e+00)(-1.250000e+01,3.025736e+00)(-1.225000e+01,3.082654e+00)(-12,3.246773e+00)(-1.175000e+01,3.331598e+00)(-1.150000e+01,3.492824e+00)(-1.125000e+01,3.463328e+00)(-11,3.663978e+00)(-1.075000e+01,3.742207e+00)(-1.050000e+01,3.876238e+00)(-1.025000e+01,4.010238e+00)(-10,4.166903e+00)(-9.750000e+00,4.220302e+00)(-9.500000e+00,4.362510e+00)(-9.250000e+00,4.475404e+00)(-9,4.672200e+00)(-8.750000e+00,4.845696e+00)(-8.500000e+00,4.979171e+00)(-8.250000e+00,5.097735e+00)(-8,5.281740e+00)(-7.750000e+00,5.488957e+00)(-7.500000e+00,5.684169e+00)(-7.250000e+00,5.803421e+00)(-7,5.994759e+00)(-6.750000e+00,6.211207e+00)(-6.500000e+00,6.376693e+00)(-6.250000e+00,6.516999e+00)(-6,6.816089e+00)(-5.750000e+00,6.961867e+00)(-5.500000e+00,7.241689e+00)(-5.250000e+00,7.519127e+00)(-5,7.778200e+00)(-4.750000e+00,7.851202e+00)(-4.500000e+00,8.181038e+00)(-4.250000e+00,8.406710e+00)(-4,8.672893e+00)(-3.750000e+00,9.017384e+00)(-3.500000e+00,9.252420e+00)(-3.250000e+00,9.499700e+00)(-3,9.856938e+00)(-2.750000e+00,9.993410e+00)(-2.500000e+00,1.025912e+01)(-2.250000e+00,1.062893e+01)(-2,1.083230e+01)(-1.750000e+00,1.098558e+01)(-1.500000e+00,1.112996e+01)(-1.250000e+00,1.122255e+01)(-1,1.140858e+01)(-7.500000e-01,1.152146e+01)(-5.000000e-01,1.146977e+01)(-2.500000e-01,1.146049e+01)(0,1.151778e+01)(2.500000e-01,1.138040e+01)(5.000000e-01,1.125085e+01)(7.500000e-01,1.106315e+01)(1,1.090326e+01)(1.250000e+00,1.068006e+01)(1.500000e+00,1.046256e+01)(1.750000e+00,1.022815e+01)(2,9.988300e+00)(2.250000e+00,9.744742e+00)(2.500000e+00,9.495623e+00)(2.750000e+00,9.248842e+00)(3,8.999324e+00)(3.250000e+00,8.749758e+00)(3.500000e+00,8.499772e+00)(3.750000e+00,8.249955e+00)(4,8.000000e+00)(4.250000e+00,7.750000e+00)(4.500000e+00,7.500000e+00)(4.750000e+00,7.250000e+00)(5,7)(5.250000e+00,6.750000e+00)(5.500000e+00,6.500000e+00)(5.750000e+00,6.250000e+00)(6,6.000000e+00)(6.250000e+00,5.750000e+00)(6.500000e+00,5.500000e+00)(6.750000e+00,5.250000e+00)(7,5)(7.250000e+00,4.750000e+00)(7.500000e+00,4.500000e+00)(7.750000e+00,4.250000e+00)(8,4.000000e+00)(8.250000e+00,3.750000e+00)(8.500000e+00,3.500000e+00)(8.750000e+00,3.250000e+00)(9,3)(9.250000e+00,2.750000e+00)(9.500000e+00,2.500000e+00)(9.750000e+00,2.250000e+00)(10,2.000000e+00)
         		};
\addlegendentry{ $a=0.4, \ \ \frac{1}{\sigma^2}=12dB$}      
 
 \addplot[color=blue,mark size=1pt,mark=point,line width=0.6pt] coordinates{
(-14,2.518376e+00)(-1.375000e+01,2.549450e+00)(-1.350000e+01,2.708652e+00)(-1.325000e+01,2.739904e+00)(-13,2.883294e+00)(-1.275000e+01,2.977756e+00)(-1.250000e+01,3.014838e+00)(-1.225000e+01,3.137883e+00)(-12,3.178854e+00)(-1.175000e+01,3.277670e+00)(-1.150000e+01,3.423050e+00)(-1.125000e+01,3.444139e+00)(-11,3.595972e+00)(-1.075000e+01,3.726461e+00)(-1.050000e+01,3.797862e+00)(-1.025000e+01,3.902100e+00)(-10,4.002292e+00)(-9.750000e+00,4.133001e+00)(-9.500000e+00,4.279429e+00)(-9.250000e+00,4.398455e+00)(-9,4.501543e+00)(-8.750000e+00,4.620981e+00)(-8.500000e+00,4.799065e+00)(-8.250000e+00,4.952847e+00)(-8,5.035063e+00)(-7.750000e+00,5.236101e+00)(-7.500000e+00,5.395563e+00)(-7.250000e+00,5.453939e+00)(-7,5.653946e+00)(-6.750000e+00,5.889167e+00)(-6.500000e+00,5.972504e+00)(-6.250000e+00,6.188782e+00)(-6,6.335983e+00)(-5.750000e+00,6.524628e+00)(-5.500000e+00,6.671722e+00)(-5.250000e+00,6.801171e+00)(-5,6.979632e+00)(-4.750000e+00,7.151328e+00)(-4.500000e+00,7.287715e+00)(-4.250000e+00,7.502036e+00)(-4,7.597134e+00)(-3.750000e+00,7.747408e+00)(-3.500000e+00,7.924422e+00)(-3.250000e+00,8.021595e+00)(-3,8.179173e+00)(-2.750000e+00,8.255300e+00)(-2.500000e+00,8.319442e+00)(-2.250000e+00,8.381126e+00)(-2,8.294298e+00)(-1.750000e+00,8.365319e+00)(-1.500000e+00,8.381192e+00)(-1.250000e+00,8.356915e+00)(-1,8.295283e+00)(-7.500000e-01,8.181694e+00)(-5.000000e-01,8.100480e+00)(-2.500000e-01,7.967044e+00)(0,7.776626e+00)(2.500000e-01,7.586562e+00)(5.000000e-01,7.380814e+00)(7.500000e-01,7.183100e+00)(1,6.955821e+00)(1.250000e+00,6.719156e+00)(1.500000e+00,6.482326e+00)(1.750000e+00,6.240717e+00)(2,5.994426e+00)(2.250000e+00,5.747062e+00)(2.500000e+00,5.499026e+00)(2.750000e+00,5.249844e+00)(3,4.999724e+00)(3.250000e+00,4.749997e+00)(3.500000e+00,4.499985e+00)(3.750000e+00,4.250000e+00)(4,4.000000e+00)(4.250000e+00,3.750000e+00)(4.500000e+00,3.500000e+00)(4.750000e+00,3.249998e+00)(5,3.000000e+00)(5.250000e+00,2.750000e+00)(5.500000e+00,2.500000e+00)(5.750000e+00,2.250000e+00)(6,2.000000e+00)(6.250000e+00,1.750000e+00)(6.500000e+00,1.500000e+00)(6.750000e+00,1.250000e+00)(7,1.000000e+00)(7.250000e+00,7.500000e-01)(7.500000e+00,5.000000e-01)(7.750000e+00,2.500000e-01)(8,0)(8.250000e+00,-2.500000e-01)(8.500000e+00,-5.000000e-01)(8.750000e+00,-7.500000e-01)(9,-1.000000e+00)(9.250000e+00,-1.250000e+00)(9.500000e+00,-1.500000e+00)(9.750000e+00,-1.750000e+00)(10,-2.000000e+00)
	};
\addlegendentry{ $a=0,\ \ \ \ \frac{1}{\sigma^2}=8dB$}  
        \addplot[color=blue,mark size=1pt,mark=x,line width=0.6pt] coordinates{
(-14,2.584258e+00)(-1.375000e+01,2.626578e+00)(-1.350000e+01,2.791202e+00)(-1.325000e+01,2.828085e+00)(-13,2.977923e+00)(-1.275000e+01,3.076630e+00)(-1.250000e+01,3.115588e+00)(-1.225000e+01,3.244470e+00)(-12,3.293421e+00)(-1.175000e+01,3.397466e+00)(-1.150000e+01,3.549968e+00)(-1.125000e+01,3.583174e+00)(-11,3.738480e+00)(-1.075000e+01,3.881576e+00)(-1.050000e+01,3.963871e+00)(-1.025000e+01,4.068855e+00)(-10,4.189059e+00)(-9.750000e+00,4.326766e+00)(-9.500000e+00,4.480546e+00)(-9.250000e+00,4.612772e+00)(-9,4.735128e+00)(-8.750000e+00,4.860786e+00)(-8.500000e+00,5.055790e+00)(-8.250000e+00,5.225450e+00)(-8,5.328578e+00)(-7.750000e+00,5.532457e+00)(-7.500000e+00,5.712373e+00)(-7.250000e+00,5.792499e+00)(-7,5.998220e+00)(-6.750000e+00,6.264475e+00)(-6.500000e+00,6.368843e+00)(-6.250000e+00,6.602559e+00)(-6,6.773950e+00)(-5.750000e+00,6.979961e+00)(-5.500000e+00,7.150016e+00)(-5.250000e+00,7.307775e+00)(-5,7.490458e+00)(-4.750000e+00,7.697029e+00)(-4.500000e+00,7.854096e+00)(-4.250000e+00,8.083519e+00)(-4,8.188522e+00)(-3.750000e+00,8.352078e+00)(-3.500000e+00,8.538301e+00)(-3.250000e+00,8.613708e+00)(-3,8.767143e+00)(-2.750000e+00,8.845173e+00)(-2.500000e+00,8.882493e+00)(-2.250000e+00,8.917552e+00)(-2,8.835019e+00)(-1.750000e+00,8.853453e+00)(-1.500000e+00,8.822715e+00)(-1.250000e+00,8.746154e+00)(-1,8.627946e+00)(-7.500000e-01,8.470806e+00)(-5.000000e-01,8.332858e+00)(-2.500000e-01,8.145321e+00)(0,7.920949e+00)(2.500000e-01,7.696050e+00)(5.000000e-01,7.462849e+00)(7.500000e-01,7.233616e+00)(1,6.991412e+00)(1.250000e+00,6.744744e+00)(1.500000e+00,6.497406e+00)(1.750000e+00,6.248751e+00)(2,5.999372e+00)(2.250000e+00,5.749725e+00)(2.500000e+00,5.499950e+00)(2.750000e+00,5.250000e+00)(3,4.999999e+00)(3.250000e+00,4.750000e+00)(3.500000e+00,4.500000e+00)(3.750000e+00,4.250000e+00)(4,4)(4.250000e+00,3.750000e+00)(4.500000e+00,3.500000e+00)(4.750000e+00,3.250000e+00)(5,3.000000e+00)(5.250000e+00,2.750000e+00)(5.500000e+00,2.500000e+00)(5.750000e+00,2.250000e+00)(6,2.000000e+00)(6.250000e+00,1.750000e+00)(6.500000e+00,1.500000e+00)(6.750000e+00,1.250000e+00)(7,1.000000e+00)(7.250000e+00,7.500000e-01)(7.500000e+00,5.000000e-01)(7.750000e+00,2.500000e-01)(8,0)(8.250000e+00,-2.500000e-01)(8.500000e+00,-5.000000e-01)(8.750000e+00,-7.500000e-01)(9,-1.000000e+00)(9.250000e+00,-1.250000e+00)(9.500000e+00,-1.500000e+00)(9.750000e+00,-1.750000e+00)(10,-2.000000e+00)

	};
\addlegendentry{ $a=0.4,\ \ \frac{1}{\sigma^2}=8dB$}     
 \addlegendimage{/pgfplots/refstyle=plot_1}\addlegendentry{  $\eta_t, \ \ a=0$}  
  \addlegendimage{/pgfplots/refstyle=plot_0}\addlegendentry{  $\eta_t, \ \ a=0.4$}

       \end{axis}
       
        \begin{axis}[
        axis y line=right,
               axis x line=none,
      xmin=-14,
      ymin=0,
      xmax=10,
      ymax=1,
      grid=major,
      scaled ticks=true,
   			xlabel={},
   			ylabel={Power Efficiency  $\eta_t$ [\%]}			
      ]

        \addplot[color=red,mark size=1pt,mark=point,line width=0.6pt] coordinates{
      
(-14,8.008070e-01)(-1.390000e+01,8.008070e-01)(-1.380000e+01,8.008070e-01)(-1.370000e+01,8.008070e-01)(-1.360000e+01,8.008070e-01)(-1.350000e+01,8.008070e-01)(-1.340000e+01,8.008070e-01)(-1.330000e+01,8.008070e-01)(-1.320000e+01,8.008070e-01)(-1.310000e+01,8.008070e-01)(-13,8.008070e-01)(-1.290000e+01,8.008070e-01)(-1.280000e+01,8.008070e-01)(-1.270000e+01,8.008070e-01)(-1.260000e+01,8.008070e-01)(-1.250000e+01,8.008070e-01)(-1.240000e+01,8.008070e-01)(-1.230000e+01,8.008070e-01)(-1.220000e+01,8.008070e-01)(-1.210000e+01,8.008070e-01)(-12,8.008070e-01)(-1.190000e+01,8.008070e-01)(-1.180000e+01,8.008070e-01)(-1.170000e+01,8.008070e-01)(-1.160000e+01,8.008070e-01)(-1.150000e+01,8.008070e-01)(-1.140000e+01,8.008070e-01)(-1.130000e+01,8.008070e-01)(-1.120000e+01,8.008070e-01)(-1.110000e+01,8.008070e-01)(-11,8.008070e-01)(-1.090000e+01,8.008070e-01)(-1.080000e+01,8.008070e-01)(-1.070000e+01,8.008070e-01)(-1.060000e+01,8.008070e-01)(-1.050000e+01,8.008070e-01)(-1.040000e+01,8.008070e-01)(-1.030000e+01,8.008070e-01)(-1.020000e+01,8.008070e-01)(-1.010000e+01,8.008070e-01)(-10,8.008070e-01)(-9.900000e+00,8.008070e-01)(-9.800000e+00,8.008070e-01)(-9.700000e+00,8.008070e-01)(-9.600000e+00,8.008070e-01)(-9.500000e+00,8.008070e-01)(-9.400000e+00,8.008070e-01)(-9.300000e+00,8.008070e-01)(-9.200000e+00,8.008070e-01)(-9.100000e+00,8.008070e-01)(-9,8.008070e-01)(-8.900000e+00,8.008070e-01)(-8.800000e+00,8.008070e-01)(-8.700000e+00,8.008070e-01)(-8.600000e+00,8.008070e-01)(-8.500000e+00,8.008070e-01)(-8.400000e+00,8.008070e-01)(-8.300000e+00,8.008070e-01)(-8.200000e+00,8.008070e-01)(-8.100000e+00,8.008070e-01)(-8,8.008070e-01)(-7.900000e+00,8.008070e-01)(-7.800000e+00,8.008070e-01)(-7.700000e+00,8.008070e-01)(-7.600000e+00,8.008070e-01)(-7.500000e+00,8.008070e-01)(-7.400000e+00,8.008070e-01)(-7.300000e+00,8.008070e-01)(-7.200000e+00,8.008070e-01)(-7.100000e+00,8.008070e-01)(-7,8.008070e-01)(-6.900000e+00,8.008070e-01)(-6.800000e+00,8.008070e-01)(-6.700000e+00,8.008070e-01)(-6.600000e+00,8.008070e-01)(-6.500000e+00,8.008070e-01)(-6.400000e+00,8.008070e-01)(-6.300000e+00,8.008070e-01)(-6.200000e+00,8.008070e-01)(-6.100000e+00,8.008070e-01)(-6,8.008070e-01)(-5.900000e+00,8.008070e-01)(-5.800000e+00,8.008070e-01)(-5.700000e+00,8.008070e-01)(-5.600000e+00,8.008070e-01)(-5.500000e+00,8.008070e-01)(-5.400000e+00,8.008070e-01)(-5.300000e+00,8.008070e-01)(-5.200000e+00,8.008070e-01)(-5.100000e+00,8.008070e-01)(-5,8.008070e-01)(-4.900000e+00,8.008070e-01)(-4.800000e+00,8.008070e-01)(-4.700000e+00,8.008070e-01)(-4.600000e+00,8.008069e-01)(-4.500000e+00,8.008069e-01)(-4.400000e+00,8.008069e-01)(-4.300000e+00,8.008069e-01)(-4.200000e+00,8.008069e-01)(-4.100000e+00,8.008069e-01)(-4,8.008069e-01)(-3.900000e+00,8.008069e-01)(-3.800000e+00,8.008069e-01)(-3.700000e+00,8.008068e-01)(-3.600000e+00,8.008068e-01)(-3.500000e+00,8.008067e-01)(-3.400000e+00,8.008065e-01)(-3.300000e+00,8.008063e-01)(-3.200000e+00,8.008060e-01)(-3.100000e+00,8.008055e-01)(-3,8.008047e-01)(-2.900000e+00,8.008036e-01)(-2.800000e+00,8.008020e-01)(-2.700000e+00,8.007997e-01)(-2.600000e+00,8.007964e-01)(-2.500000e+00,8.007916e-01)(-2.400000e+00,8.007850e-01)(-2.300000e+00,8.007756e-01)(-2.200000e+00,8.007626e-01)(-2.100000e+00,8.007446e-01)(-2,8.007200e-01)(-1.900000e+00,8.006866e-01)(-1.800000e+00,8.006416e-01)(-1.700000e+00,8.005814e-01)(-1.600000e+00,8.005017e-01)(-1.500000e+00,8.003970e-01)(-1.400000e+00,8.002604e-01)(-1.300000e+00,8.000837e-01)(-1.200000e+00,7.998570e-01)(-1.100000e+00,7.995686e-01)(-1,7.992046e-01)(-9.000000e-01,7.987488e-01)(-8.000000e-01,7.981829e-01)(-7.000000e-01,7.974857e-01)(-6.000000e-01,7.966339e-01)(-5.000000e-01,7.956015e-01)(-4.000000e-01,7.943602e-01)(-3.000000e-01,7.928797e-01)(-2.000000e-01,7.911278e-01)(-1.000000e-01,7.890711e-01)(0,7.866756e-01)(1.000000e-01,7.839069e-01)(2.000000e-01,7.807317e-01)(3.000000e-01,7.771180e-01)(4.000000e-01,7.730361e-01)(5.000000e-01,7.684596e-01)(6.000000e-01,7.633663e-01)(7.000000e-01,7.577387e-01)(8.000000e-01,7.515646e-01)(9.000000e-01,7.448383e-01)(1,7.375602e-01)(1.100000e+00,7.297374e-01)(1.200000e+00,7.213837e-01)(1.300000e+00,7.125192e-01)(1.400000e+00,7.031701e-01)(1.500000e+00,6.933684e-01)(1.600000e+00,6.831505e-01)(1.700000e+00,6.725571e-01)(1.800000e+00,6.616317e-01)(1.900000e+00,6.504203e-01)(2,6.389697e-01)(2.100000e+00,6.273271e-01)(2.200000e+00,6.155390e-01)(2.300000e+00,6.036505e-01)(2.400000e+00,5.917045e-01)(2.500000e+00,5.797412e-01)(2.600000e+00,5.677977e-01)(2.700000e+00,5.559077e-01)(2.800000e+00,5.441012e-01)(2.900000e+00,5.324047e-01)(3,5.208408e-01)(3.100000e+00,5.094289e-01)(3.200000e+00,4.981849e-01)(3.300000e+00,4.871218e-01)(3.400000e+00,4.762497e-01)(3.500000e+00,4.655762e-01)(3.600000e+00,4.551068e-01)(3.700000e+00,4.448451e-01)(3.800000e+00,4.347930e-01)(3.900000e+00,4.249512e-01)(4,4.153193e-01)(4.100000e+00,4.058958e-01)(4.200000e+00,3.966786e-01)(4.300000e+00,3.876651e-01)(4.400000e+00,3.788524e-01)(4.500000e+00,3.702369e-01)(4.600000e+00,3.618151e-01)(4.700000e+00,3.535833e-01)(4.800000e+00,3.455376e-01)(4.900000e+00,3.376742e-01)(5,3.299891e-01)(5.100000e+00,3.224786e-01)(5.200000e+00,3.151387e-01)(5.300000e+00,3.079657e-01)(5.400000e+00,3.009558e-01)(5.500000e+00,2.941054e-01)(5.600000e+00,2.874109e-01)(5.700000e+00,2.808687e-01)(5.800000e+00,2.744754e-01)(5.900000e+00,2.682276e-01)(6,2.621220e-01)(6.100000e+00,2.561554e-01)(6.200000e+00,2.503246e-01)(6.300000e+00,2.446265e-01)(6.400000e+00,2.390581e-01)(6.500000e+00,2.336165e-01)(6.600000e+00,2.282987e-01)(6.700000e+00,2.231020e-01)(6.800000e+00,2.180236e-01)(6.900000e+00,2.130608e-01)(7,2.082109e-01)(7.100000e+00,2.034715e-01)(7.200000e+00,1.988399e-01)(7.300000e+00,1.943137e-01)(7.400000e+00,1.898906e-01)(7.500000e+00,1.855682e-01)(7.600000e+00,1.813441e-01)(7.700000e+00,1.772162e-01)(7.800000e+00,1.731823e-01)(7.900000e+00,1.692402e-01)(8,1.653878e-01)(8.100000e+00,1.616231e-01)(8.200000e+00,1.579441e-01)(8.300000e+00,1.543489e-01)(8.400000e+00,1.508355e-01)(8.500000e+00,1.474021e-01)(8.600000e+00,1.440468e-01)(8.700000e+00,1.407679e-01)(8.800000e+00,1.375636e-01)(8.900000e+00,1.344323e-01)(9,1.313722e-01)(9.100000e+00,1.283818e-01)(9.200000e+00,1.254595e-01)(9.300000e+00,1.226037e-01)(9.400000e+00,1.198129e-01)(9.500000e+00,1.170856e-01)(9.600000e+00,1.144204e-01)(9.700000e+00,1.118159e-01)(9.800000e+00,1.092706e-01)(9.900000e+00,1.067833e-01)(10,1.043527e-01)

       };\label{plot_0}
       
       \addplot[color=red,mark size=0.5pt,mark=triangle,line width=0.6pt] coordinates{
       (-14,8.000000e-01)(-1.390000e+01,8.000000e-01)(-1.380000e+01,8.000000e-01)(-1.370000e+01,8.000000e-01)(-1.360000e+01,8.000000e-01)(-1.350000e+01,8.000000e-01)(-1.340000e+01,8.000000e-01)(-1.330000e+01,8.000000e-01)(-1.320000e+01,8.000000e-01)(-1.310000e+01,8.000000e-01)(-13,8.000000e-01)(-1.290000e+01,8.000000e-01)(-1.280000e+01,8.000000e-01)(-1.270000e+01,8.000000e-01)(-1.260000e+01,8.000000e-01)(-1.250000e+01,8.000000e-01)(-1.240000e+01,8.000000e-01)(-1.230000e+01,8.000000e-01)(-1.220000e+01,8.000000e-01)(-1.210000e+01,8.000000e-01)(-12,8.000000e-01)(-1.190000e+01,8.000000e-01)(-1.180000e+01,8.000000e-01)(-1.170000e+01,8.000000e-01)(-1.160000e+01,8.000000e-01)(-1.150000e+01,8.000000e-01)(-1.140000e+01,8.000000e-01)(-1.130000e+01,8.000000e-01)(-1.120000e+01,8.000000e-01)(-1.110000e+01,8.000000e-01)(-11,8.000000e-01)(-1.090000e+01,8.000000e-01)(-1.080000e+01,8.000000e-01)(-1.070000e+01,8.000000e-01)(-1.060000e+01,8.000000e-01)(-1.050000e+01,8.000000e-01)(-1.040000e+01,8.000000e-01)(-1.030000e+01,8.000000e-01)(-1.020000e+01,8.000000e-01)(-1.010000e+01,8.000000e-01)(-10,8.000000e-01)(-9.900000e+00,8.000000e-01)(-9.800000e+00,8.000000e-01)(-9.700000e+00,8.000000e-01)(-9.600000e+00,8.000000e-01)(-9.500000e+00,8.000000e-01)(-9.400000e+00,8.000000e-01)(-9.300000e+00,8.000000e-01)(-9.200000e+00,8.000000e-01)(-9.100000e+00,8.000000e-01)(-9,8.000000e-01)(-8.900000e+00,8.000000e-01)(-8.800000e+00,8.000000e-01)(-8.700000e+00,8.000000e-01)(-8.600000e+00,8.000000e-01)(-8.500000e+00,8.000000e-01)(-8.400000e+00,8.000000e-01)(-8.300000e+00,8.000000e-01)(-8.200000e+00,8.000000e-01)(-8.100000e+00,8.000000e-01)(-8,8.000000e-01)(-7.900000e+00,8.000000e-01)(-7.800000e+00,8.000000e-01)(-7.700000e+00,8.000000e-01)(-7.600000e+00,8.000000e-01)(-7.500000e+00,8.000000e-01)(-7.400000e+00,8.000000e-01)(-7.300000e+00,8.000000e-01)(-7.200000e+00,8.000000e-01)(-7.100000e+00,8.000000e-01)(-7,8.000000e-01)(-6.900000e+00,8.000000e-01)(-6.800000e+00,8.000000e-01)(-6.700000e+00,8.000000e-01)(-6.600000e+00,8.000000e-01)(-6.500000e+00,8.000000e-01)(-6.400000e+00,8.000000e-01)(-6.300000e+00,8.000000e-01)(-6.200000e+00,8.000000e-01)(-6.100000e+00,8.000000e-01)(-6,8.000000e-01)(-5.900000e+00,8.000000e-01)(-5.800000e+00,8.000000e-01)(-5.700000e+00,8.000000e-01)(-5.600000e+00,8.000000e-01)(-5.500000e+00,8.000000e-01)(-5.400000e+00,8.000000e-01)(-5.300000e+00,8.000000e-01)(-5.200000e+00,8.000000e-01)(-5.100000e+00,8.000000e-01)(-5,8.000000e-01)(-4.900000e+00,8.000000e-01)(-4.800000e+00,8.000000e-01)(-4.700000e+00,7.999999e-01)(-4.600000e+00,7.999999e-01)(-4.500000e+00,7.999998e-01)(-4.400000e+00,7.999997e-01)(-4.300000e+00,7.999996e-01)(-4.200000e+00,7.999993e-01)(-4.100000e+00,7.999990e-01)(-4,7.999984e-01)(-3.900000e+00,7.999975e-01)(-3.800000e+00,7.999963e-01)(-3.700000e+00,7.999944e-01)(-3.600000e+00,7.999917e-01)(-3.500000e+00,7.999877e-01)(-3.400000e+00,7.999819e-01)(-3.300000e+00,7.999736e-01)(-3.200000e+00,7.999619e-01)(-3.100000e+00,7.999455e-01)(-3,7.999225e-01)(-2.900000e+00,7.998908e-01)(-2.800000e+00,7.998474e-01)(-2.700000e+00,7.997884e-01)(-2.600000e+00,7.997090e-01)(-2.500000e+00,7.996031e-01)(-2.400000e+00,7.994631e-01)(-2.300000e+00,7.992797e-01)(-2.200000e+00,7.990416e-01)(-2.100000e+00,7.987352e-01)(-2,7.983444e-01)(-1.900000e+00,7.978504e-01)(-1.800000e+00,7.972316e-01)(-1.700000e+00,7.964635e-01)(-1.600000e+00,7.955183e-01)(-1.500000e+00,7.943656e-01)(-1.400000e+00,7.929725e-01)(-1.300000e+00,7.913034e-01)(-1.200000e+00,7.893212e-01)(-1.100000e+00,7.869877e-01)(-1,7.842644e-01)(-9.000000e-01,7.811133e-01)(-8.000000e-01,7.774981e-01)(-7.000000e-01,7.733851e-01)(-6.000000e-01,7.687445e-01)(-5.000000e-01,7.635510e-01)(-4.000000e-01,7.577855e-01)(-3.000000e-01,7.514351e-01)(-2.000000e-01,7.444942e-01)(-1.000000e-01,7.369651e-01)(0,7.288577e-01)(1.000000e-01,7.201897e-01)(2.000000e-01,7.109865e-01)(3.000000e-01,7.012800e-01)(4.000000e-01,6.911086e-01)(5.000000e-01,6.805156e-01)(6.000000e-01,6.695484e-01)(7.000000e-01,6.582575e-01)(8.000000e-01,6.466948e-01)(9.000000e-01,6.349128e-01)(1,6.229634e-01)(1.100000e+00,6.108970e-01)(1.200000e+00,5.987614e-01)(1.300000e+00,5.866013e-01)(1.400000e+00,5.744578e-01)(1.500000e+00,5.623678e-01)(1.600000e+00,5.503642e-01)(1.700000e+00,5.384754e-01)(1.800000e+00,5.267257e-01)(1.900000e+00,5.151355e-01)(2,5.037212e-01)(2.100000e+00,4.924961e-01)(2.200000e+00,4.814702e-01)(2.300000e+00,4.706508e-01)(2.400000e+00,4.600430e-01)(2.500000e+00,4.496499e-01)(2.600000e+00,4.394728e-01)(2.700000e+00,4.295118e-01)(2.800000e+00,4.197659e-01)(2.900000e+00,4.102331e-01)(3,4.009110e-01)(3.100000e+00,3.917963e-01)(3.200000e+00,3.828859e-01)(3.300000e+00,3.741758e-01)(3.400000e+00,3.656623e-01)(3.500000e+00,3.573414e-01)(3.600000e+00,3.492090e-01)(3.700000e+00,3.412612e-01)(3.800000e+00,3.334940e-01)(3.900000e+00,3.259032e-01)(4,3.184851e-01)(4.100000e+00,3.112357e-01)(4.200000e+00,3.041513e-01)(4.300000e+00,2.972280e-01)(4.400000e+00,2.904623e-01)(4.500000e+00,2.838507e-01)(4.600000e+00,2.773894e-01)(4.700000e+00,2.710753e-01)(4.800000e+00,2.649049e-01)(4.900000e+00,2.588749e-01)(5,2.529822e-01)(5.100000e+00,2.472236e-01)(5.200000e+00,2.415961e-01)(5.300000e+00,2.360967e-01)(5.400000e+00,2.307225e-01)(5.500000e+00,2.254706e-01)(5.600000e+00,2.203383e-01)(5.700000e+00,2.153228e-01)(5.800000e+00,2.104214e-01)(5.900000e+00,2.056317e-01)(6,2.009509e-01)(6.100000e+00,1.963767e-01)(6.200000e+00,1.919066e-01)(6.300000e+00,1.875383e-01)(6.400000e+00,1.832694e-01)(6.500000e+00,1.790977e-01)(6.600000e+00,1.750209e-01)(6.700000e+00,1.710370e-01)(6.800000e+00,1.671437e-01)(6.900000e+00,1.633390e-01)(7,1.596210e-01)(7.100000e+00,1.559876e-01)(7.200000e+00,1.524369e-01)(7.300000e+00,1.489670e-01)(7.400000e+00,1.455761e-01)(7.500000e+00,1.422624e-01)(7.600000e+00,1.390241e-01)(7.700000e+00,1.358595e-01)(7.800000e+00,1.327670e-01)(7.900000e+00,1.297448e-01)(8,1.267915e-01)(8.100000e+00,1.239053e-01)(8.200000e+00,1.210849e-01)(8.300000e+00,1.183287e-01)(8.400000e+00,1.156352e-01)(8.500000e+00,1.130030e-01)(8.600000e+00,1.104307e-01)(8.700000e+00,1.079170e-01)(8.800000e+00,1.054605e-01)(8.900000e+00,1.030600e-01)(9,1.007140e-01)(9.100000e+00,9.842150e-02)(9.200000e+00,9.618115e-02)(9.300000e+00,9.399180e-02)(9.400000e+00,9.185229e-02)(9.500000e+00,8.976148e-02)(9.600000e+00,8.771826e-02)(9.700000e+00,8.572154e-02)(9.800000e+00,8.377028e-02)(9.900000e+00,8.186344e-02)(10,8.000000e-02)
       };\label{plot_1}

       \end{axis}
       
  \end{tikzpicture} \vskip-3mm
\centering
  \caption{SINR at the left Y-axis and $\eta_t$ at the right Y-axis vs. $\frac{\frac{1}{c}-1}{\gamma}$ when $K=40$ and $M=80$.}
  \label{fig1}
\end{figure}

\begin{figure}
  \centering
   \begin{tikzpicture}[scale=0.8,font=\small]
    \renewcommand{\axisdefaulttryminticks}{4}
    \tikzstyle{every major grid}+=[style=densely dashed]
    \tikzstyle{every axis y label}+=[yshift=0pt]
    \tikzstyle{every axis x label}+=[yshift=0pt]
    \tikzstyle{every axis legend}+=[cells={anchor=west},fill=white,
        at={(0.02,0.98)}, anchor=north west, font=\tiny ]
    \begin{axis}[
      xmin=80,
      ymin=10,
      xmax=120,
      ymax=16,
      grid=major,
      scaled ticks=true,
   			xlabel={Number of BS antennas $M$},
   			ylabel={SINR [dB] }			
      ]
  
          \addplot[color=blue,mark size=1.5pt,mark=triangle,line width=0.8pt] coordinates{
(80,1.249008e+01)(85,1.293078e+01)(90,1.333381e+01)(95,1.370469e+01)(100,1.404790e+01)(105,1.436708e+01)(110,1.466526e+01)(115,1.494493e+01)(120,1.520819e+01)          		};
\addlegendentry{ Optimal $\gamma$}

                   \addplot[color=green,mark size=1.5pt,mark=star,line width=0.8pt] coordinates{    
(80,1.068883e+01)(85,1.133276e+01)(90,1.196018e+01)(95,1.264657e+01)(100,1.316673e+01)(105,1.367022e+01)(110,1.405118e+01)(115,1.437364e+01)(120,1.460691e+01)
      };
                              \addlegendentry{ $\gamma=2$}      
 \addplot[color=magenta,mark size=1.5pt,mark=x,line width=0.8pt] coordinates{    
(80,1.188621e+01)(85,1.242308e+01)(90,1.287670e+01)(95,1.325851e+01)(100,1.345966e+01)(105,1.360530e+01)(110,1.368105e+01)(115,1.372680e+01)(120,1.374497e+01)
         };
                              \addlegendentry{$\gamma=1.5$}
       \end{axis}
  \end{tikzpicture} \vskip-3mm
\centering
  \caption{SINR vs. $M$ when $K=40$, $\frac{1}{\sigma^2}=12dB$ and $\bR=\bI_M$.}
  \label{fig2}
\end{figure}

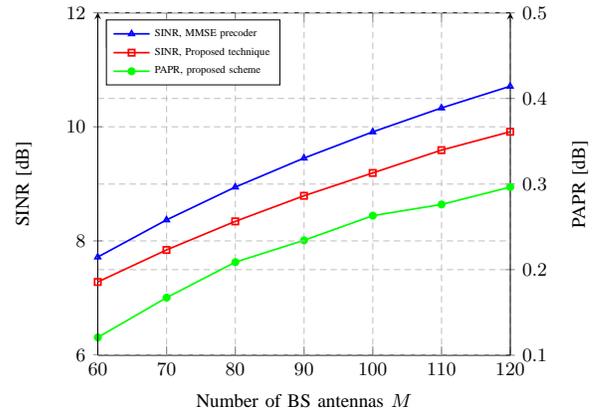
\begin{figure}
  \centering
   \begin{tikzpicture}[scale=0.8,font=\small]
    \renewcommand{\axisdefaulttryminticks}{4}
    \tikzstyle{every major grid}+=[style=densely dashed]
    \tikzstyle{every axis y label}+=[xshift=0pt]
    \tikzstyle{every axis x label}+=[yshift=0pt]
    \tikzstyle{every axis legend}+=[cells={anchor=west},fill=white,
        at={(0.02,0.98)}, anchor=north west, font=\tiny ]
    \begin{axis}[
    axis y line=right,
      xmin=60,
      ymin=0.1,
      xmax=120,
      ymax=0.5,
      grid=major,
      scaled ticks=true,
   			xlabel={Number of BS antennas $M$},
   			ylabel={PAPR [dB] }			
      ]
         \addplot[color=green,mark size=1.5pt,mark=*,line width=0.8pt] coordinates{
(60,1.208627e-01)(70,1.673120e-01)(80,2.087753e-01)(90,2.343027e-01)(100,2.630976e-01)(110,2.761807e-01)(120,2.964630e-01)(130,3.333905e-01)(140,3.326868e-01)
     		};\label{plot_one}
\addlegendentry{ PAPR, proposed scheme}                                                           
              \end{axis}
              
               \begin{axis}[
               axis y line=left,
               axis x line=none,
                 xmin=60, xmax=120,
           ymin=6,
          ymax=12,
      grid=major,
   			ylabel={SINR [dB] }			
      ]
 
          \addplot[color=blue,mark size=1.5pt,mark=triangle,line width=0.8pt] coordinates{
(60,7.716543e+00)(70,8.367558e+00)(80,8.942825e+00)(90,9.450849e+00)(100,9.910114e+00)(110,1.032979e+01)(120,1.071041e+01)

   		};\addlegendentry{ SINR, MMSE precoder}

  \addplot[color=red,mark size=1.5pt,mark=square,line width=0.8pt] coordinates{

(60,7.280289e+00)(70,7.841417e+00)(80,8.342152e+00)(90,8.790768e+00)(100,9.191022e+00)(110,9.591098e+00)(120,9.912887e+00)

    		};\addlegendentry{ SINR, Proposed technique}

     \addlegendimage{/pgfplots/refstyle=plot_one}\addlegendentry{ PAPR, proposed scheme}                    
              \end{axis}

  \end{tikzpicture} \vskip-3mm
\centering
  \caption{PAPR at the right Y-axis and SINR at the left Y-axis vs. $M$ when $K=10$, $\frac{1}{\sigma^2}=0dB$ and $a=0.1$.}
  \label{fig3}
\end{figure}

\section{Conclusion}
In this paper, we have studied the performance of a recently proposed precoding scheme for single-RF massive MIMO system. Unlike classical methods of precoding design, a more realistic model for power amplifiers is used by considering a peak power constraint. Using tools from random matrix theory, we have derived deterministic approximations of the optimal parameters of the considered precoding technique. It was shown that this precoding scheme has low PAPR and good SINR compared to the classical MMSE precoder.

\section*{Appendix A}
\label{appA}
Starting from \eqref{interference}, 
\begin{align*}
\frac{1}{K}\| \bH^{\mbox{\tiny H}} \bx - \sqrt{\gamma}\bu\|^2&=\frac{\gamma{\rho}^2}{K}\bu^{\mbox{\tiny H}} \left(\frac{\bH^{\mbox{\tiny H}} \bH}{K}+\rho \bI\right)^{-2}  \bu\\
&=\rho\left(\frac{\gamma}{K}  \bu^{\mbox{\tiny H}} \left(\frac{\bH^{\mbox{\tiny H}} \bH}{K}+\rho \bI\right)^{-1}  \bu -P(\rho)\right).
\end{align*}
We have an asymptotic equivalent of $P(\rho)$, it remains to deal with the term $\frac{1}{K}  \bu^{\mbox{\tiny H}} \left(\frac{\bH^{\mbox{\tiny H}} \bH}{K}+\rho \bI\right)^{-1}\bu$. It is known that such quantity converges to its mean. Exploiting this result,
\begin{align*}
&\frac{1}{K} \mathbb{E} \bu^{\mbox{\tiny H}} \left(\frac{\bH^{\mbox{\tiny H}} \bH}{K}+\rho \bI_K\right)^{-1}\bu\\&=\frac{1}{K} \mathbb{E} \tr \left(\frac{\bH^{\mbox{\tiny H}} \bH}{K}+\rho \bI\right)^{-1}\\
&=\frac{1}{K\rho} \mathbb{E} \tr \left(\frac{\bH^{\mbox{\tiny H}} \bH}{\rho K}+ \bI_K\right)^{-1}\\
&=\frac{1}{K\rho} \left[\mathbb{E} \tr \left(\frac{\bH\bH^{\mbox{\tiny H}} }{\rho K}+ \bI_M\right)^{-1}-\frac{(M-K)}{K}\right].\\
\end{align*}
It is known from \cite[Theorem 1]{Kammoun2014a} that for $t>0$,
$$
 \frac{1}{K}\tr \left(\frac{t\bH \bH^{\mbox{\tiny H}}}{ K}+ \bI_M\right)^{-1}- \frac{1}{K}\tr \left(\bT(t)\right)\xrightarrow[K\to \infty]{a.s.}0.
$$
Putting all the above results together, we have
\begin{align*}
&\frac{1}{K}\| \bH^{\mbox{\tiny H}} \bx- \sqrt{\gamma}\bu\|^2-\\&\left(\frac{\gamma}{K}\left(\tr \left[\bT(\frac{1}{\rho})\right]+M-K\right)-\rho \overline P_1(\rho)\right) \xrightarrow[K\to \infty]{a.s.}0.
\end{align*}

\section*{Appendix B}
\label{appB}
To determine the optimal $\gamma$, it suffices to determine the optimal $\rho$ that minimizes the following function
$$
\tilde g(\rho)=\frac{\sigma^2}{P_a\rho^2}\beta(\frac{1}{\rho})+\left(\frac{1}{K}\tr \left(\bT(\frac{1}{\rho})\right)-\frac{M-K}{K}-\frac{1}{\rho} \beta(\frac{1}{\rho})\right).
$$
From fig. 1, we conclude that there exists a unique $\gamma$, and thus a unique $\rho$, that maximize the SINR. Then, it suffices to determine $\rho$ that satisfies $\frac{\partial}{\partial \rho}\tilde g(\rho)=0$.
One can easily show that 
$$
\frac{1}{K}\frac{\partial}{\partial t}\tr \left(\bT(t)\right)=-\beta(t), \ \ \ \forall t>0.
$$
Thus, 
$$
\tilde g'(\rho)=\left(\frac{-2\sigma^2}{P_a\rho^3}+\frac{2}{\rho^2}\right)\beta(\frac{1}{\rho})+\left(\frac{-\sigma^2}{P_a\rho^4}+\frac{1}{\rho^3}\right)\beta'(\frac{1}{\rho}).$$
Then, it can be checked that $g'(\frac{\sigma^2}{P_a})=0$. \\
Thus, $\rho^*=\frac{\sigma^2}{P_a}$ and $\gamma^*=\frac{P_a}{\frac{P_a^2}{\sigma^4}\beta(\frac{P_a}{\sigma^2})}=\frac{1}{\frac{P_a}{\sigma^4}\beta(\frac{P_a}{\sigma^2})}$.

\section*{Appendix C}
\label{appC}
By replacing $\delta$ by $\rho=\frac{\delta}{K}$ in \eqref{P}, we have
\begin{align*}
P(\rho)&=\frac{\gamma}{K} \bu^{\mbox{\tiny H}} \left(\frac{\bH^{\mbox{\tiny H}} \bH}{K}+\rho \bI\right)^{-1}\frac{ \bH^{\mbox{\tiny H}} \bH}{K} \left(\frac{\bH^{\mbox{\tiny H}} \bH}{K}+  \rho \bI\right)^{-1} \bu\\
&=\frac{\gamma}{K}  \bu^{\mbox{\tiny H}} \left(\frac{\bH^{\mbox{\tiny H}} \bH}{K}+\rho \bI\right)^{-1}  \bu -\frac{\gamma\rho}{K} \bu^{\mbox{\tiny H}}\left(\frac{\bH^{\mbox{\tiny H}} \bH}{K}+\rho \bI\right)^{-2} \bu\\
\end{align*}
Let $\bQ(\rho)= \left(\frac{\bH^{\mbox{\tiny H}} \bH}{K}+\rho \bI\right)^{-1}$.
From \cite{COUbook}, we have
$$
\frac{1}{K}  \bu^{\mbox{\tiny H}} \bQ(\rho)  \bu -m(\rho)\xrightarrow[K\to \infty]{a.s.}0
$$
where $$m(\rho)=\frac{-2}{1-c-\rho-\sqrt{(1-c+\rho)^2+4c\rho}} .$$
And by noting that 
$$
\frac{1}{K} \bu^{\mbox{\tiny H}} \left(\frac{\bH^{\mbox{\tiny H}} \bH}{K}+\rho \bI\right)^{-2}  \bu=-\frac{\partial}{\partial\rho}\left[\frac{1}{K}  \bu^H \bQ(\rho)  \bu\right]
$$
Then, we have
\begin{equation}
\frac{1}{K} \bu^{\mbox{\tiny H}} \left(\frac{\bH^{\mbox{\tiny H}} \bH}{K}+\rho \bI\right)^{-2}  \bu+m'(\rho)\xrightarrow[K\to \infty]{a.s.}0,
\label{conver_deriv}
\end{equation}
where $m'(\rho)$ is the first derivative of $m(\rho)$ with respect to $\rho$.
Putting the above results together, we have
$$
P(\rho)-\gamma\left[m(\rho)+\rho m(\rho)\right]\xrightarrow[K\to \infty]{a.s.}0.
$$

\section*{Appendix D}
\label{appD}
The SINR can be written as
$$
 {\rm SINR} =\frac{\frac{\gamma}{K} \bu^{\mbox{\tiny H}}\bu}{\frac{1}{K}\| \bH^{\mbox{\tiny H}} \bx - \sqrt{\gamma}\bu\|^2+\sigma^2}
$$
First, it is known that
$$
\frac{1}{K} \bu^{\mbox{\tiny H}}\bu-1\xrightarrow[K\to \infty]{a.s.}0,
$$
It remains to deal with the interference term,
\begin{align*}
\frac{1}{K}\| \bH^{\mbox{\tiny H}} \bx - \sqrt{\gamma}\bu\|^2=(\bx-\bb)^{\mbox{\tiny H}}\bA(\bx-\bb)
\end{align*}
After some algebraic manipulations, we have
\begin{equation}
\frac{1}{K}\| \bH^{\mbox{\tiny H}} \bx - \sqrt{\gamma}\bu\|^2=\frac{\gamma{\rho}^2}{K}\bu^{\mbox{\tiny H}} \left(\frac{\bH^{\mbox{\tiny H}} \bH}{K}+\rho \bI\right)^{-2}  \bu.\label{interference}
\end{equation}
Thus, using \eqref{conver_deriv},
\begin{align*}
\frac{1}{K}\| \bH^{\mbox{\tiny H}} \bx - \sqrt{\gamma}\bu\|^2+\gamma{\rho}^2 m'(\rho)\xrightarrow[K\to \infty]{a.s.}0,
\end{align*}
Putting the above results together yields,
$$
{\rm SINR}-\frac{1}{-\rho^2 m'(\rho)+\frac{\sigma^2}{\gamma}}\xrightarrow[K\to \infty]{a.s.}0.
$$
Replacing $\gamma$ by $\frac{P_a}{\bar f(\rho)}$ yields the result of proposition \ref{prop2}.

\bibliographystyle{IEEEbib}
\bibliography{IEEEabrv,IEEEconf,tutorial_RMT}

\end{document}